\documentclass[journal]{IEEEtran}
\usepackage{subfigure}
\usepackage{amsmath,epsfig,latexsym}
\usepackage{multicol}
\usepackage{bm}
\usepackage{subfigure}
\usepackage{amssymb}
\usepackage{amsfonts}
\usepackage{graphicx}
\usepackage{url}
\usepackage{ccaption}
\usepackage{booktabs} 
\usepackage{color}
\usepackage{multirow}
\usepackage{amsthm}

\newtheorem{remark}{Remark}

\newtheorem{definition}{Definition}

\newtheorem{theorem}{Theorem}

\ifCLASSINFOpdf
\else
\fi
\begin{document}
%
\title{Robust Student's t based Stochastic Cubature Filter for Nonlinear Systems with Heavy-tailed Process and Measurement Noises}

%
%
%

\author{Yulong~Huang,
        Yonggang~Zhang,~\IEEEmembership{Senior Member,~IEEE,}

\thanks{This work was supported by the National Natural Science Foundation of China under Grant Nos. 61371173 and 61633008 and the Natural Science Foundation of Heilongjiang Province Grant No. F2016008. Corresponding author is Y. G. Zhang.}

\thanks{Y. L. Huang and Y. G. Zhang are with the Department of Automation, Harbin Engineering University, Harbin 150001, China (e-mail: heuedu@163.com; zhangyg@hrbeu.edu.cn).}}

\maketitle

\begin{abstract}
In this paper, a new robust Student's t based stochastic cubature filter (RSTSCF) is proposed for nonlinear state-space model with heavy-tailed process and measurement noises. The heart of the RSTSCF is a stochastic Student's t spherical radial cubature rule (SSTSRCR), which is derived based on the third-degree unbiased spherical rule and the proposed third-degree unbiased radial rule. The existing stochastic integration rule is a special case of the proposed SSTSRCR when the degrees of freedom parameter tends to infinity. The proposed filter is applied to a manoeuvring bearings-only tracking example, where an agile target is tracked and the bearing is observed in clutter. Simulation results show that the proposed RSTSCF can achieve higher estimation accuracy than the existing Gaussian approximate filter, Gaussian sum filter, Huber-based nonlinear Kalman filter, maximum correntropy criterion based Kalman filter, and robust Student's t based nonlinear filters, and is computationally much more efficient than the existing particle filter.
\end{abstract}


%
\IEEEpeerreviewmaketitle

\section{Introduction}
\IEEEPARstart{N}{onlinear} filtering has been playing an important role in many applications, such as target tracking, detection, signal processing, communication and navigation. Under the Bayesian estimation framework, the nonlinear filtering problem is addressed by calculating the posterior probability density function (PDF) recursively based on the nonlinear state-space model. Unfortunately, there is not a closed form solution for posterior PDF for nonlinear state-space model since a closed PDF for nonlinear mapping doesn't exist \cite{Anderson79}. As a result, there is not an optimal solution for nonlinear filtering problem, and an approximate approach is necessary to obtain a suboptimal solution. In general, the posterior PDF is approximated as Gaussian by assuming the jointly predicted PDF of the state and measurement vectors is Gaussian, and the resultant Gaussian approximate (GA) filter can provide tradeoffs between estimation accuracy and computational complexity \cite{CGAF 2015,Arasaratnam 2009}. Up to present, many variants of the GA filter have been proposed using different Gaussian weighted integral rules \cite{Arasaratnam 2009}--\cite{ECKF 2015}. However, in some engineering applications, such as tracking an agile target that is observed in clutter, the heavy-tailed process noise may be induced by severe manoeuvering and the heavy-tailed measurement noise may be induced by measurement outliers from unreliable sensors \cite{student's t filter 2013,Robust GA smoother 2016,TAES 2017 RSTKF}. The performance of the GA filters may degrade for such engineering applications with heavy-tailed noises since they all model the process and measurement noises as Gaussian distributions so that they are sensitive to heavy-tailed non-Gaussian noises \cite{Robust GA smoother 2016}.

Particle filter (PF) is a common method to address non-Gaussian noises, in which the posterior PDF is approximated as a set of random samples with associated weights based on sequential Monte Carlo sampling technique \cite{Arulampalam 2002}. The PF can model the process and measurement noises as arbitrary distributions, such as the Student's t distributions for heavy-tailed non-Gaussian noises \cite{STPF 1993,STPF 2006}. However, the PF suffers from substantial computational complexities in high-dimensional problems because the number of particles increases exponentially with the dimensionality of the state \cite{visual tracking 2008}. Gaussian sum filter (GSF) is an alternative method to handle heavy-tailed non-Gaussian noises, where the heavy-tailed process and measurement noises are modelled as a finite sum of Gaussian distributions, and the posterior distribution is then approximated as a weighted sum of Gaussian distributions by running a bank of GA filters \cite{GSF 1972}--\cite{GSF 2010}. However, for the GSF, it is very difficult to model the heavy-tailed process and measurement noises accurately using finite Gaussian distributions since the heavy-tailed non-Gaussian noises are induced by the unknown manoeuvering or outliers, which may degrade the estimation performance of the GSF.

To solve the filtering problem of nonlinear state-space model with heavy-tailed non-Gaussian noises, the Huber-based nonlinear Kalman filter (HNKF) has been proposed by minimising a Huber cost function that is a combined $l_{1}$ and $l_{2}$ norm \cite{Nonlinear Huber estimator 2010}. A larger number of variants of the HNKF have been derived based on a linearized or statistical linearized method, such as the Huber-based extended Kalman filter \cite{Huber estimator 1995}, the Huber-based divided difference filter \cite{Huber estimator 2007}, the Huber-based unscented Kalman filter \cite{Nonlinear Huber estimator 2012}, the nonlinear regression Huber Kalman filter \cite{Nonlinear regression HKF 2015} and the adaptively robust unscented Kalman filter (ARUKF) \cite{ARUKF 2014}. However, the influence function of the HNKFs don't redescend, which may deteriorate the estimation performance of the HNKFs \cite{TAES 2017 RSTKF}. The maximum correntropy criterion based Kalman filter (MCCKF) has been proposed by maximising the correntropy of the predicted error and residual \cite{MCCKF 2012b}--\cite{MCCKF 2016}. However, there is a lack of theoretical basis to develop the estimation error covariance matrix of the MCCKF, which may degrade the estimation accuracy \cite{TAES 2017 RSTKF}.

A reasonable approach to improve the estimation performance is utilizing a Student's t distribution to model the heavy-tailed non-Gaussian noise. The Student's t distribution is a generalized Gaussian distribution but has heavier tails than the Gaussian distribution, which makes it more suitable for modelling the heavy-tailed non-Gaussian noise. A general framework of the robust Student's t based nonlinear filter (RSTNF) has been proposed, in which the jointly predicted PDF of the state and measurement vectors is assumed to be Student's t, and the posterior PDF is then approximated as Student's t \cite{TAES 2016}. The heart of the RSTNF is how to calculate the Student's t weighted integral, and the estimation accuracy of the associated RSTNF is determined by the employed numerical integral technique. Many variants of the RSTNF have been derived based on different numerical integral methods, such as the robust Student's t based extended filter (RSTEF) using the first-order linearization \cite{student's t filter 2013}, the robust Student's t based unscented filter (RSTUF) using the unscented transform (UT) \cite{TAES 2016,Fusion 2016 Sarkka}, and the robust Student's t based cubature filter (RSTCF) using the third-degree Student's t spherical radial cubature rule (STSRCR) \cite{Fusion 2016}. However, the existing Student's t integral rules can only capture the third-degree or fifth-degree information of the Taylor series expansion for nonlinear approximation, which may result in limited estimation accuracy. Although the Monte Carlo approach can be used to calculate the Student's t weighted integral, it has low accuracy and slow convergence when the integrand is not approximately constant and the number of random samples is finite \cite{Genz 1998}. Therefore, there is a great demand to develop more accurate numerical integral approach for the Student's t weighted integral to further improve the estimation accuracy of the existing RSTNFs.

In this paper, the Student's t weighted integral is decomposed into the spherical integral and the radial integral based on the spherical-radial transformation. A new stochastic STSRCR (SSTSRCR) is derived based on the third-degree unbiased spherical rule (USR) and the proposed third-degree unbiased radial rule (URR), from which a new robust Student's t based stochastic cubature filter (RSTSCF) is obtained. The existing stochastic integration rule (SIR) \cite{Dunik13} is a special case of the proposed SSTSRCR when the degrees of freedom (dof) parameter tends to infinity. The proposed SSTSRCR can achieve better approximation to the Student's t weighted integral as compared with existing Student's t integral rules. As a result, the proposed RSTSCF has higher estimation accuracy than the existing RSTNFs. The proposed filter and existing filters are applied to a manoeuvring bearings-only tracking example, where an agile target is tracked and the bearing is observed in clutter. Simulation results show that the proposed RSTSCF can achieve higher estimation accuracy than the existing GA filter, GSF, HNKF, MCCKF and RSTNFs, and is computationally much more efficient than the existing PF.

The remainder of this paper is organized as follows. In Section II, a general frame of the RSTNF is reviewed. In Section III, a new SSTSRCR is derived based on the proposed third-degree URR, from which a new RSTSCF is obtained. Also, the relationship between the proposed SSTSRCR and the existing SIR is revealed in Section III. In Section IV, the proposed filter is applied to a manoeuvring bearings-only tracking example and simulation results are given. Concluding remarks are drawn in Section V.

\section{Problem Statement}
Consider the following discrete-time nonlinear stochastic system as represented by the state-space model \cite{TAES 2016}
\begin{equation}
\mathbf{x}_{k}=\mathbf{f}_{k-1}(\mathbf{x}_{k-1})+\mathbf{w}_{k-1} \qquad \textrm{(process equation)}
\end{equation}
\begin{equation}
\mathbf{z}_{k}=\mathbf{h}_{k}(\mathbf{x}_{k})+\mathbf{v}_{k} \qquad\;\; \textrm{(measurement equation)},
\end{equation}
where ${k}$ is the discrete time index, $\mathbf{x}_k\in\mathbb{R}^{n}$ is the state vector, $\mathbf{z}_k\in\mathbb{R}^{m}$ is the measurement vector, and $\mathbf{f}_{k-1}(\cdot)$ and $\mathbf{h}_{k}(\cdot)$ are known process and measurement functions respectively. $\mathbf{w}_k\in\mathbb{R}^{n}$ and $\mathbf{v}_k\in\mathbb{R}^{m}$ are heavy-tailed process and measurement noise vectors respectively, which are induced by process and measurement outliers, and their distributions are modelled as Student's t distributions, i.e.,
\begin{equation}
p(\mathbf{w}_{k})=\mathrm{St}(\mathbf{w}_{k};\mathbf{0},\mathbf{Q}_{k},\nu_{1})
\end{equation}
\begin{equation}
p(\mathbf{v}_{k})=\mathrm{St}(\mathbf{v}_{k};\mathbf{0},\mathbf{R}_{k},\nu_{2}),
\end{equation}
where $\mathrm{St}(\cdot;\mathbf{\mu},\mathbf{\Sigma},v)$ denotes the Student's t PDF with mean vector $\mathbf{\mu}$, scale matrix $\mathbf{\Sigma}$, and dof parameter $v$, $\mathbf{Q}_{k}$ and $\nu_{1}$ are the scale matrix and dof parameter of process noise respectively, and $\mathbf{R}_{k}$ and $\nu_{2}$ are the scale matrix and dof parameter of measurement noise respectively. The initial state vector $\mathbf{x}_{0}$ is also assumed to have a Student's t distribution with mean vector $\mathbf{\hat{x}}_{0|0}$, scale matrix $\mathbf{P}_{0|0}$, and dof parameter $\nu_{3}$, and $\mathbf{x}_{0}$, $\mathbf{w}_{k}$ and $\mathbf{v}_{k}$ are assumed to be mutually uncorrelated.

To achieve the filtering estimation, a general framework of RSTNF is derived for the nonlinear system formulated in equations (1)-(4), where the jointly predicted PDF of the state and measurement vectors is assumed as Student's t, then the posterior PDF of the state vector can be approximated as Student's t \cite{TAES 2016}. The time update and measurement update of the recursive RSTNF are given as follows: \\
\textbf{Time update}
\begin{eqnarray}
&&\mathbf{\hat{x}}_{k|k-1}=\int_{\mathbb{R}^{n}} \mathbf{f}_{k-1}(\mathbf{x}_{k-1})\mathrm{St}(\mathbf{x}_{k-1};\mathbf{\hat{x}}_{k-1|k-1},\mathbf{P}_{k-1|k-1},\nu_{3}) \nonumber\\
&&d\mathbf{x}_{k-1}
\end{eqnarray}
\begin{eqnarray}
&&\mathbf{P}_{k|k-1}=\frac{\nu_{3}-2}{\nu_{3}}\int_{\mathbb{R}^{n}} \mathbf{f}_{k-1}(\mathbf{x}_{k-1})\mathbf{f}_{k-1}^{T}(\mathbf{x}_{k-1})\mathrm{St}(\mathbf{x}_{k-1}; \nonumber\\
&&\mathbf{\hat{x}}_{k-1|k-1},\mathbf{P}_{k-1|k-1},\nu_{3})d\mathbf{x}_{k-1}-\frac{\nu_{3}-2}{\nu_{3}}\mathbf{\hat{x}}_{k|k-1}\mathbf{\hat{x}}_{k|k-1}^{T} \nonumber\\
&&+\frac{\nu_{1}(\nu_{3}-2)}{(\nu_{1}-2)\nu_{3}}\mathbf{Q}_{k-1},
\end{eqnarray}
where $(\cdot)^{T}$ denotes the transpose operation, $\mathbf{\hat{x}}_{k|k-1}$ and $\mathbf{P}_{k|k-1}$ are respectively the mean vector and scale matrix of the one-step predicted PDF $p(\mathbf{x}_{k}|\mathbf{Z}_{k-1})$, $\mathbf{Z}_{k-1}=\{\mathbf{z}_{j}\}_{j=1}^{k-1}$ is the set of $k-1$ measurement vectors, and $\nu_{3}$ denotes the dof parameter of the filtering PDF. \\
\textbf{Measurement update}
\begin{equation}
\Delta_{k}=\sqrt{(\mathbf{z}_{k}-\mathbf{\hat{z}}_{k|k-1})^{T}(\mathbf{P}_{k|k-1}^{zz})^{-1}(\mathbf{z}_{k}-\mathbf{\hat{z}}_{k|k-1})}
\end{equation}
\begin{equation}
\mathbf{K}_{k}=\mathbf{P}_{k|k-1}^{xz}(\mathbf{P}_{k|k-1}^{zz})^{-1}
\end{equation}
\begin{equation}
\mathbf{\hat{x}}_{k|k}=\mathbf{\hat{x}}_{k|k-1}+\mathbf{K}_{k}(\mathbf{z}_{k}-\mathbf{\hat{z}}_{k|k-1})
\end{equation}
\begin{equation}
\mathbf{P}_{k|k}=\frac{(\nu_{3}-2)(\nu_{3}+\Delta_{k}^{2})}{\nu_{3}(\nu_{3}+m-2)}(\mathbf{P}_{k|k-1}-\mathbf{K}_{k}\mathbf{P}_{k|k-1}^{zz}\mathbf{K}_{k}^{T}),
\end{equation}
where $(\cdot)^{\mathrm{-1}}$ denotes the inverse operation, $\mathbf{\hat{x}}_{k|k}$ and $\mathbf{P}_{k|k}$ are respectively the mean vector and scale matrix of the filtering PDF $p(\mathbf{x}_{k}|\mathbf{Z}_{k})$, $\mathbf{\hat{z}}_{k|k-1}$ and $\mathbf{P}_{k|k-1}^{zz}$ are respectively the mean vector and scale matrix of the likelihood PDF $p(\mathbf{z}_{k}|$ $\mathbf{Z}_{k-1})$, and $\mathbf{P}_{k|k-1}^{xz}$ is the cross scale matrix of state and measurement vectors, which are given by
\begin{eqnarray}
\mathbf{\hat{z}}_{k|k-1}&=&\int_{\mathbb{R}^{n}} \mathbf{h}_{k}(\mathbf{x}_{k})\mathrm{St}(\mathbf{x}_{k};\mathbf{\hat{x}}_{k|k-1},\mathbf{P}_{k|k-1},\nu_{3})d\mathbf{x}_{k}
\end{eqnarray}
\begin{eqnarray}
&&\mathbf{P}_{k|k-1}^{zz}=\frac{\nu_{3}-2}{\nu_{3}}\int_{\mathbb{R}^{n}} \mathbf{h}_{k}(\mathbf{x}_{k})\mathbf{h}_{k}^{T}(\mathbf{x}_{k})\mathrm{St}(\mathbf{x}_{k};\mathbf{\hat{x}}_{k|k-1},\mathbf{P}_{k|k-1},
 \nonumber\\
&&\nu_{3})d\mathbf{x}_{k}-\frac{\nu_{3}-2}{\nu_{3}}\mathbf{\hat{z}}_{k|k-1}\mathbf{\hat{z}}_{k|k-1}^{T}+\frac{\nu_{2}(\nu_{3}-2)}{(\nu_{2}-2)\nu_{3}}\mathbf{R}_{k}
\end{eqnarray}
\begin{eqnarray}
&&\mathbf{P}_{k|k-1}^{xz}=\frac{\nu_{3}-2}{\nu_{3}}\int_{\mathbb{R}^{n}} \mathbf{x}_{k}\mathbf{h}_{k}^{T}(\mathbf{x}_{k})\mathrm{St}(\mathbf{x}_{k};\mathbf{\hat{x}}_{k|k-1},\mathbf{P}_{k|k-1},\nu_{3}) \nonumber\\
&&d\mathbf{x}_{k}-\frac{\nu_{3}-2}{\nu_{3}}\mathbf{\hat{x}}_{k|k-1}\mathbf{\hat{z}}_{k|k-1}^{T}.
\end{eqnarray}

The recursive RSTNF is composed of the analytical computations in equations (7)-(10) and the Student's t weighted integrals in equations (5)-(6) and (11)-(13). The key problem in the design of the RSTNF is calculating the nonlinear Student's t weighted integrals formulated in equations (5)-(6) and (11)-(13), whose integrands are all of the form \emph{nonlinear function}$\times$\emph{Student's t PDF}. Therefore, the numerical integral technique is required to implement the RSTNF, which determines the estimation accuracy of associated RSTNF. Next, to further improve the estimation accuracy of existing RSTNFs, a new SSTSRCR will be proposed, based on which a new RSTSCF can be obtained.

\section{Main Results}
\subsection{Spherical-radial transformation}
The Student's t weighted integrals involved in the RSTNF can be written as the general form as follows
\begin{equation}
I[\mathbf{g}]=\int_{\mathbb{R}^{n}}\mathbf{g}(\mathbf{x})\mathrm{St}(\mathbf{x};\mathbf{\mu},\mathbf{\Sigma},\nu)d\mathbf{x},
\end{equation}
where the Student's t PDF is given by
\begin{eqnarray}
&&\mathrm{St}(\mathbf{x};\mathbf{\mu},\mathbf{\Sigma},\nu)=\frac{\Gamma(\frac{\nu+n}{2})}{\Gamma(\frac{\nu}{2})}\frac{1}{\sqrt{|\nu\pi\mathbf{\Sigma}|}}\times \nonumber\\
&&\left[1+\frac{1}{\nu}(\mathbf{x}-\mathbf{\mu})^{\mathrm{T}}\mathbf{\Sigma}^{-1}(\mathbf{x}-\mathbf{\mu})\right]^{-\frac{\nu+n}{2}},
\end{eqnarray}
where $\Gamma(\cdot)$ and $|\cdot|$ denote the Gamma function and determinant operation respectively. To derive the SSTSRCR, the Student's t weighted integral in equation (14) requires to be transformed into a spherical-radial integral form.

A change of variable is utilized as follows
\begin{equation}
\mathbf{x}=\mathbf{\mu}+\sqrt{\nu\mathbf{\Sigma}}\mathbf{y},
\end{equation}
where $\sqrt{\mathbf{\Sigma}}$ is the square-root of scale matrix $\mathbf{\Sigma}$ satisfying $\mathbf{\Sigma}=\sqrt{\mathbf{\Sigma}}\sqrt{\mathbf{\Sigma}}^{T}$.

Substituting equation (16) in equations (14)-(15) and using the identity $|\sqrt{\nu\mathbf{\Sigma}}|=\sqrt{|\nu\mathbf{\Sigma}|}$ yields
\begin{equation}
I[\mathbf{g}]=\int_{\mathbb{R}^{n}}\mathbf{l}(\mathbf{y})(1+\mathbf{y}^{\mathrm{T}}\mathbf{y})^{-\frac{\nu+n}{2}}d\mathbf{y},
\end{equation}
where $\mathbf{l}(\mathbf{y})$ is given by
\begin{equation}
\mathbf{l}(\mathbf{y})=\frac{\Gamma(\frac{\nu+n}{2})}{\Gamma(\frac{\nu}{2})\pi^{\frac{n}{2}}}\mathbf{g}(\mathbf{\mu}+\sqrt{\nu\mathbf{\Sigma}}\mathbf{y}).
\end{equation}

Define $\mathbf{y}=r\mathbf{s}$ with $\mathbf{s}^{T}\mathbf{s}=1$, then equation (17) can be rewritten as \cite{Stroud 1971}
\begin{eqnarray}
I[\mathbf{g}]&=&\int_{0}^{+\infty}\int_{U_{n}}\mathbf{l}(r\mathbf{s})r^{n-1}[1+(r\mathbf{s})^{T}(r\mathbf{s})]^{-\frac{\nu+n}{2}}d\sigma(s)dr \nonumber\\
&=&\int_{0}^{+\infty}\int_{U_{n}}\mathbf{l}(r\mathbf{s})r^{n-1}(1+r^2)^{-\frac{\nu+n}{2}}d\sigma(s)dr,
\end{eqnarray}
where $\mathbf{s}=[\mathbf{s}_{1},\mathbf{s}_{2},\cdots,\mathbf{s}_{n}]^{T}$, $U_{n}=\{\mathbf{s}\in\mathbb{R}^{n}:s_{1}^{2}+s_{2}^{2}+\cdots+s_{n}^{2}=1\}$, and $\sigma(\mathbf{s})$ is the spherical surface measure or an area element on $U_{n}$.

According to equation (19), the Student's t weighted integral in equation (14) can be decomposed into the radial integral
\begin{equation}
I[\mathbf{g}]=\int_{0}^{+\infty}\mathbf{S}(r)r^{n-1}(1+r^2)^{-\frac{\nu+n}{2}}dr,
\end{equation}
and the spherical integral
\begin{equation}
\mathbf{S}(r)=\int_{U_{n}}\mathbf{l}(r\mathbf{s})d\sigma(s).
\end{equation}

Next, a new third-degree SSTSRCR will be derived, in which the spherical and the radial integrals are respectively calculated by the third-degree USR (Section III. B below) and the third-degree URR (Section III. C below). Before deriving the third-degree SSTSRCR, the unbiased integral rule is firstly defined as follows.
\begin{definition}
The integral rule $\int\mathbf{g}(\mathbf{x})p(\mathbf{x})d\mathbf{x}\approx\sum\limits_{l=1}^{N}w_{l}\mathbf{g}(\mathbf{x}_{l})$ is unbiased if and only if \cite{Genz 1998}
\begin{equation}
\int\mathbf{g}(\mathbf{x})p(\mathbf{x})d\mathbf{x}=E\left[\sum\limits_{l=1}^{N}w_{l}\mathbf{g}(\mathbf{x}_{l})\right],
\end{equation}
where $\mathbf{x}_{l}$ and $w_{l}$ are respectively cubature points and corresponding weights, and $E[\cdot]$ denotes the expectation operation.
\end{definition}

\subsection{Unbiased spherical rule}
The Stewart's method is employed to construct the third-degree USR. If $\mathbf{Q}$ is a random orthogonal matrix drawn with a Haar distribution from the set of all matrices in the orthogonal group, the third-degree USR can be constructed as \cite{Genz 1998,Stewart 1980}
\begin{equation}
\mathbf{S}_{u}^{3}(r)\approx\frac{A_{n}}{2n}\sum\limits_{i=1}^{n}\left[\mathbf{l}(-r\mathbf{Q}\mathbf{e}_{i})+\mathbf{l}(r\mathbf{Q}\mathbf{e}_{i})\right],
\end{equation}
where $A_{n}=\frac{2\pi^{\frac{n}{2}}}{\Gamma(\frac{n}{2})}$ is the surface area of the unit sphere, and $\mathbf{e}_{i}$ denotes the $i$-th column of an $n\times{n}$ unit matrix. To produce a random orthogonal matrix $\mathbf{Q}$, a $n\times n$ matrix $\mathbf{U}$ of standard norm variables is first generated, then the required random orthogonal matrix $\mathbf{Q}$ is obtained based on the $QR$ factorization, i.e., $\mathbf{U}=\mathbf{Q}\mathbf{R}$ \cite{Stewart 1980}.

Next, a new third-degree URR will be proposed for the radial integral in equation (20).

\subsection{Unbiased radial rule}
Generally, the monomials $S(r)=1$, $S(r)=r$, $S(r)=r^{2}$, and $S(r)=r^{3}$ need to be matched to derive the third-degree URR. However, only monomials $S(r)=1$ and $S(r)=r^{2}$ need to be matched for the third-degree URR since the USR and the resultant STSRCR are fully symmetry. Thus, two points $\{r_{1}, \omega_{r,1}\}$ and $\{r_{2}, \omega_{r,2}\}$ are sufficient to design the third-degree URR, where one point is used to match monomials $S(r)=1$ and $S(r)=r^{2}$ and the other point is employed to retain unbiasedness. That is to say, the third-degree URR can be written as
\begin{equation}
\int_{0}^{+\infty}\mathbf{S}(r)r^{n-1}(1+r^2)^{-\frac{\nu+n}{2}}dr\approx\omega_{r,1}\mathbf{S}(r_{1})+\omega_{r,2}\mathbf{S}(r_{2}),
\end{equation}
where $\{r_{1}, \omega_{r,1}\}$ and $\{r_{2}, \omega_{r,2}\}$ satisfy the following equations
\begin{equation}
\omega_{r,1}r_{1}^{0}+\omega_{r,2}r_{2}^{0}=\int_{0}^{+\infty}r^{0}r^{n-1}(1+r^2)^{-\frac{\nu+n}{2}}dr
\end{equation}
\begin{equation}
\omega_{r,1}r_{1}^{2}+\omega_{r,2}r_{2}^{2}=\int_{0}^{+\infty}r^{2}r^{n-1}(1+r^2)^{-\frac{\nu+n}{2}}dr
\end{equation}
\begin{equation}
\int_{0}^{+\infty}\mathbf{S}(r)r^{n-1}(1+r^2)^{-\frac{\nu+n}{2}}dr=E\left[\omega_{r,1}\mathbf{S}(r_{1})+\omega_{r,2}\mathbf{S}(r_{2})\right].
\end{equation}

Since there are three equations and four variables in equations (25)-(27), there is one free variable. In order to get the STSRCR with the minimum number of points, $r_{1}$ is chosen as the free variable and set to zero.

\begin{theorem}
If $r_{1}=0$ and the PDF of random variable $r_{2}$ is
\begin{equation}
p(r_{2})=2r_{2}^{n+1}(1+r_{2}^{2})^{-\frac{\nu+n}{2}}/\mathrm{B}(\frac{n+2}{2},\frac{\nu-2}{2}),
\end{equation}
where $\mathrm{B}(\cdot,\cdot)$ denotes the beta function, then the third-degree URR is given by
{\small\begin{equation}
I[\mathbf{g}]\approx\frac{1}{2}\mathrm{B}(\frac{n}{2},\frac{\nu}{2})
\left\{\left[1-\frac{n}{(\nu-2)r_{2}^{2}}\right]\mathbf{S}(0)+\frac{n}{(\nu-2)r_{2}^{2}}\mathbf{S}(r_{2})\right\}.
\end{equation}}
\end{theorem}

\begin{proof}
Firstly, a general integral $\int_{0}^{+\infty}r^{l}r^{n-1}(1+r^2)^{-\frac{\nu+n}{2}}dr$ is calculated to obtain the right-hand parts in equations (25)-(26). Making a change of variable via $t=r^{2}$ results in
\begin{equation}
\int_{0}^{+\infty}r^{l}r^{n-1}(1+r^2)^{-\frac{\nu+n}{2}}dr=\frac{1}{2}\mathrm{B}(\frac{n+l}{2},\frac{\nu-l}{2}),
\end{equation}
where $\mathrm{B}(\cdot,\cdot)$ denotes the beta function.

Substituting equation (30) in equations (25)-(26), we have
\begin{equation}
\omega_{r,1}+\omega_{r,2}=\frac{1}{2}\mathrm{B}(\frac{n}{2},\frac{\nu}{2})
\end{equation}
\begin{equation}
\omega_{r,1}r_{1}^{2}+\omega_{r,2}r_{2}^{2}=\frac{1}{2}\mathrm{B}(\frac{n+2}{2},\frac{\nu-2}{2}).
\end{equation}

Utilizing the identities $\Gamma(a+1)=a\Gamma(a)$ and $\mathrm{B}(a,b)=\frac{\Gamma(a)\Gamma(b)}{\Gamma(a+b)}$ in equation (32) yields
\begin{equation}
\omega_{r,1}r_{1}^{2}+\omega_{r,2}r_{2}^{2}=\frac{n}{2(\nu-2)}\mathrm{B}(\frac{n}{2},\frac{\nu}{2}).
\end{equation}

Employing $r_{1}=0$ in equation (33) yields
\begin{equation}
\omega_{r,2}=\frac{n}{2(\nu-2)r_{2}^{2}}\mathrm{B}(\frac{n}{2},\frac{\nu}{2}).
\end{equation}

Substituting equation (34) in equation (31) results in
\begin{equation}
\omega_{r,1}=\frac{1}{2}\mathrm{B}(\frac{n}{2},\frac{\nu}{2})\left[1-\frac{n}{(\nu-2)r_{2}^{2}}\right].
\end{equation}

Utilizing $r_{1}=0$ and equations (34)-(35), the expectation of the third-degree radial rule with respect to $p(r_{2})$ is written as
{\small\begin{eqnarray}
&&E\left[\omega_{r,1}\mathbf{S}(r_{1})+\omega_{r,2}\mathbf{S}(r_{2})\right]=\frac{1}{2}\mathrm{B}(\frac{n}{2},\frac{\nu}{2})E\left[1-\frac{n}{(\nu-2)r_{2}^{2}}\right]\times \nonumber\\
&&\mathbf{S}(0)+\frac{1}{2}\mathrm{B}(\frac{n}{2},\frac{\nu}{2})E\left[\frac{n}{(\nu-2)r_{2}^{2}}\mathbf{S}(r_{2})\right].
\end{eqnarray}}

Using equations (28) and (30), we have
\begin{eqnarray}
&&E\left[1-\frac{n}{(\nu-2)r_{2}^{2}}\right]=\int_{0}^{+\infty}\frac{2r_{2}^{n+1}(1+r_{2}^{2})^{-\frac{\nu+n}{2}}}{\mathrm{B}(\frac{n+2}{2},\frac{\nu-2}{2})}dr_{2}- \nonumber\\
&&\frac{n}{(\nu-2)}\int_{0}^{+\infty}\frac{2r_{2}^{n-1}(1+r_{2}^{2})^{-\frac{\nu+n}{2}}}{\mathrm{B}(\frac{n+2}{2},\frac{\nu-2}{2})}dr_{2}=0
\end{eqnarray}
{\small\begin{eqnarray}
&&E\left[\frac{n}{(\nu-2)r_{2}^{2}}\mathbf{S}(r_{2})\right]=\frac{n}{(\nu-2)}\int_{0}^{+\infty}\frac{2r_{2}^{n-1}(1+r_{2}^{2})^{-\frac{\nu+n}{2}}}{\mathrm{B}(\frac{n+2}{2},\frac{\nu-2}{2})}\times  \nonumber\\
&&\mathbf{S}(r_{2})dr_{2}=\frac{2}{\mathrm{B}(\frac{n}{2},\frac{\nu}{2})}\int_{0}^{+\infty}\mathbf{S}(r)r^{n-1}(1+r^2)^{-\frac{\nu+n}{2}}dr.
\end{eqnarray}}

Substituting equations (37)-(38) in equation (36) yields
\begin{equation}
E\left[\omega_{r,1}\mathbf{S}(r_{1})+\omega_{r,2}\mathbf{S}(r_{2})\right]=\int_{0}^{+\infty}\mathbf{S}(r)r^{n-1}(1+r^2)^{-\frac{\nu+n}{2}}dr.
\end{equation}

With $r_{1}=0$, equations (34)-(35) and (39), the third-degree URR can be formulated as equation (29).
\end{proof}
It is very difficult to directly generate random samples from $p(r_{2})$ since $p(r_{2})$ is not a special PDF. To solve this problem, Theorem 2 is presented as follows.
\begin{theorem}
If random variable $\tau=\frac{r_{2}^{2}}{1+r_{2}^{2}}$, then random variable $\tau$ obeys the Beta distribution, i.e.,
\begin{equation}
p(\tau)=\mathrm{Beta}(\tau;\frac{n+2}{2},\frac{\nu-2}{2})=\frac{\tau^{\frac{n+2}{2}-1}(1-\tau)^{\frac{\nu-2}{2}-1}}{\mathrm{B}(\frac{n+2}{2},\frac{\nu-2}{2})},
\end{equation}
where $\mathrm{Beta}(\cdot;\alpha,\beta)$ denotes the beta PDF with parameters $\alpha$ and $\beta$.
\end{theorem}

\begin{proof}
Since $\tau=\frac{r_{2}^{2}}{1+r_{2}^{2}}$ and $r_{2}\in[0,+\infty)$, random variable $\tau\in[0,1)$. According to $\tau=\frac{r_{2}^{2}}{1+r_{2}^{2}}$, $r_{2}$ is formulated as
\begin{equation}
r_{2}=c(\tau)=\sqrt{\frac{\tau}{1-\tau}} \qquad \tau\in[0,1).
\end{equation}

Employing the transformation theorem and equation (41), the PDF of random variable $\tau$ is given by
\begin{equation}
p(\tau)=p_{r_{2}}(c(\tau))c'(\tau),
\end{equation}
where $p_{r_{2}}(\cdot)$ denotes the PDF of $r_{2}$ and $c'(\tau)$ denotes the derivative of $c(\tau)$ with respect to $\tau$ given by
\begin{equation}
c'(\tau)=0.5\tau^{-\frac{1}{2}}(1-\tau)^{-\frac{3}{2}}.
\end{equation}

Substituting equations (28), (41) and (43) in equation (42) obtains
\begin{equation}
p(\tau)=\tau^{\frac{n+2}{2}-1}(1-\tau)^{\frac{\nu-2}{2}-1}/\mathrm{B}(\frac{n+2}{2},\frac{\nu-2}{2}),
\end{equation}
which proves the theorem.
\end{proof}

\subsection{Stochastic STSRCR}
A theorem is first presented to derive the unbiased STSRCR.
\begin{theorem}
If the spherical and radial rules are unbiased, then the resultant STSRCR is also unbiased.
\end{theorem}

\begin{proof}
If the spherical and radial rules are given by
\begin{equation}
\mathbf{S}(r)\approx\sum\limits_{i=1}^{N_{s}}w_{s,i}\mathbf{l}(r\mathbf{s}_{i}); \qquad
I[\mathbf{g}]\approx\sum\limits_{j=1}^{N_{r}}w_{r,j}\mathbf{S}(r_{j}),
\end{equation}
then the STSRCR can be formulated as
\begin{equation}
I[\mathbf{g}]\approx\sum\limits_{j=1}^{N_{r}}\sum\limits_{i=1}^{N_{s}}w_{r,j}w_{s,i}\mathbf{l}(r_{j}\mathbf{s}_{i}),
\end{equation}
where $\mathbf{s}_{i}$ and $w_{s,i}$ are respectively cubature points and weights of the spherical rule, and $r_{j}$ and $w_{r,j}$ are respectively quadrature points and weights of the radial rule.

Since the spherical and radial rules are unbiased, we obtain
\begin{eqnarray}
\mathbf{S}(r)=E\left[\sum\limits_{i=1}^{N_{s}}w_{s,i}\mathbf{l}(r\mathbf{s}_{i})\right]; \; I[\mathbf{g}]=E\left[\sum\limits_{j=1}^{N_{r}}w_{r,j}\mathbf{S}(r_{j})\right].
\end{eqnarray}

Using equation (47) yields
\begin{equation}
I[\mathbf{g}]=E\left\{\sum\limits_{j=1}^{N_{r}}w_{r,j}E\left[\sum\limits_{i=1}^{N_{s}}w_{s,i}\mathbf{l}(r_{j}\mathbf{s}_{i})\right]\right\}.
\end{equation}

Since the set $\{\mathbf{s}_{i},\,w_{s,i}\}_{i=1}^{N_s}$ is independent of the set $\{\mathbf{r}_{j},\,w_{r,j}\}_{j=1}^{N_r}$, we have
\begin{equation}
I[\mathbf{g}]=E\left[\sum\limits_{j=1}^{N_{r}}\sum\limits_{i=1}^{N_{s}}w_{r,j}w_{s,i}\mathbf{l}(r_{j}\mathbf{s}_{i})\right],
\end{equation}
which proves the theorem.
\end{proof}

Using Theorems 1-3 obtains
\begin{eqnarray}
&&I[\mathbf{g}]=E\left\{\left[1-\frac{n}{(\nu-2)r_{2}^{2}}\right]\mathbf{g}(\mathbf{\mu})+\frac{1}{2(\nu-2)r_{2}^{2}}\times\right. \nonumber\\
&&\left.\sum\limits_{i=1}^{n}\left[\mathbf{g}(\mathbf{\mu}-r_{2}\sqrt{\nu\mathbf{\Sigma}}\mathbf{Q}\mathbf{e}_{i})+\mathbf{g}(\mathbf{\mu}+r_{2}\sqrt{\nu\mathbf{\Sigma}}\mathbf{Q}\mathbf{e}_{i})\right]\right\}.
\end{eqnarray}

By employing the Monte Carlo approach, the right-hand parts of equation (50) can be approximated as
\begin{eqnarray}
&&I_{s}^{3}[\mathbf{g}]=\frac{1}{N}\sum\limits_{l=1}^{N}\left\{\left[1-\frac{n}{(\nu-2)r_{2,l}^{2}}\right]\mathbf{g}(\mathbf{\mu})+\frac{1}{2(\nu-2)r_{2,l}^{2}}\times\right. \nonumber\\
&&\left.\sum\limits_{i=1}^{n}\left[\mathbf{g}(\mathbf{\mu}-r_{2,l}\sqrt{\nu\mathbf{\Sigma}}\mathbf{Q}_{l}\mathbf{e}_{i})+\mathbf{g}(\mathbf{\mu}+r_{2,l}\sqrt{\nu\mathbf{\Sigma}}\mathbf{Q}_{l}\mathbf{e}_{i})\right]\right\},
\end{eqnarray}
where $N$ denotes the number of random samples, and $\mathbf{Q}_{l}$ is a random orthogonal matrix, and $r_{2,l}$ is drawn randomly from $p(r_{2})$. The form $I_{s}^{3}[\mathbf{g}]$ denotes the proposed third-degree SSTSRCR, and the implementation pseudocode of the proposed SSTSRCR is shown in Table I.
\begin{table}[!t]
\renewcommand{\arraystretch}{2.0}
\caption{The implementation pseudocode of the proposed SSTSRCR.}
\centering
\begin{tabular}[width=0.8in]{l}
\\
\hline
{\bfseries Inputs}: $\mathbf{\mu}$, $\mathbf{\Sigma}$, $\nu$, $\mathbf{g}(\cdot)$, $n$, $N$. \\
1. Initialization: $I_{s}^{3}[\mathbf{g}]=0$. \\
{\bfseries for} $l=1:N$ \\
2. Generate a $n\times n$ matrix $\mathbf{U}_{l}$ of standard norm variables. \\
3. Obtain the required random orthogonal matrix $\mathbf{Q}_{l}$ using the \\
$QR$ factorization: $\mathbf{U}_{l}=\mathbf{Q}_{l}\mathbf{R}_{l}$. \\
4. Draw the random variable $\tau_{l}$ from a Beta distribution: \\
$\tau_{l}\sim\mathrm{Beta}(\frac{n+2}{2},\frac{\nu-2}{2})$. \\
5. Calculate the random quadrature point $r_{2,l}$: \\
$r_{2,l}=\sqrt{\frac{\tau_{l}}{1-\tau_{l}}}$. \\
6. Update $I_{s}^{3}[\mathbf{g}]$ at current iteration: \\
$I_{s}^{3}[\mathbf{g}]=I_{s}^{3}[\mathbf{g}]+\frac{1}{N}\left\{\left[1-\frac{n}{(\nu-2)r_{2,l}^{2}}\right]\mathbf{g}(\mathbf{\mu})+\frac{1}{2(\nu-2)r_{2,l}^{2}}\times\right.$ \\
$\left.\sum\limits_{i=1}^{n}\left[\mathbf{g}(\mathbf{\mu}-r_{2,l}\sqrt{\nu\mathbf{\Sigma}}\mathbf{Q}_{l}\mathbf{e}_{i})+\mathbf{g}(\mathbf{\mu}+r_{2,l}\sqrt{\nu\mathbf{\Sigma}}\mathbf{Q}_{l}\mathbf{e}_{i})\right]\right\}$. \\
{\bfseries end for} \\
{\bfseries Outputs}: $I[\mathbf{g}]\approx I_{s}^{3}[\mathbf{g}]$. \\
\hline
\end{tabular}
\end{table}
\begin{table}[!t]
\renewcommand{\arraystretch}{2.0}
\caption{The implementation pseudocode for one time step of the proposed RSTSCF.}
\centering
\begin{tabular}[width=0.8in]{l}
\\
\hline
{\bfseries Inputs}: $\mathbf{\hat{x}}_{k-1|k-1}$, $\mathbf{P}_{k-1|k-1}$, $\mathbf{z}_{k}$, $\mathbf{Q}_{k-1}$, $\mathbf{R}_{k}$, $\nu_{1}$, $\nu_{2}$, $\nu_{3}$, $\mathbf{f}_{k-1}(\cdot)$, \\
$\mathbf{h}_{k}(\cdot)$, $n$, $N$. \\
{\bfseries Time update:} \\
1. $\mathbf{\hat{x}}_{k|k-1}=\mathrm{SSTSRCR}(\mathbf{\hat{x}}_{k-1|k-1},\mathbf{P}_{k-1|k-1},\nu_{3},\mathbf{f}_{k-1}(\cdot),n,N)$. \\
2. $\mathbf{P}_{k|k-1}=\frac{\nu_{3}-2}{\nu_{3}}\mathrm{SSTSRCR}(\mathbf{\hat{x}}_{k-1|k-1},\mathbf{P}_{k-1|k-1},\nu_{3},$ \\
$\mathbf{f}_{k-1}(\cdot)\mathbf{f}_{k-1}^{T}(\cdot),n,N)-\frac{\nu_{3}-2}{\nu_{3}}\mathbf{\hat{x}}_{k|k-1}\mathbf{\hat{x}}_{k|k-1}^{T}+\frac{\nu_{1}(\nu_{3}-2)}{(\nu_{1}-2)\nu_{3}}\mathbf{Q}_{k-1}$. \\
{\bfseries Measurement update:} \\
3. $\mathbf{\hat{z}}_{k|k-1}=\mathrm{SSTSRCR}(\mathbf{\hat{x}}_{k|k-1},\mathbf{P}_{k|k-1},\nu_{3},\mathbf{h}_{k}(\cdot),n,N)$. \\
4. $\mathbf{P}_{k|k-1}^{zz}=\frac{\nu_{3}-2}{\nu_{3}}\mathrm{SSTSRCR}(\mathbf{\hat{x}}_{k|k-1},\mathbf{P}_{k|k-1},\nu_{3},\mathbf{h}_{k}(\cdot)\mathbf{h}_{k}^{T}(\cdot),$ \\
$n,N)-\frac{\nu_{3}-2}{\nu_{3}}\mathbf{\hat{z}}_{k|k-1}\mathbf{\hat{z}}_{k|k-1}^{T}+\frac{\nu_{2}(\nu_{3}-2)}{(\nu_{2}-2)\nu_{3}}\mathbf{R}_{k}$. \\
5. $\mathbf{P}_{k|k-1}^{xz}=\frac{\nu_{3}-2}{\nu_{3}}\mathrm{SSTSRCR}(\mathbf{\hat{x}}_{k|k-1},\mathbf{P}_{k|k-1},\nu_{3},\mathbf{x}_{k}\mathbf{h}_{k}(\cdot),n,N)$ \\
$-\frac{\nu_{3}-2}{\nu_{3}}\mathbf{\hat{x}}_{k|k-1}\mathbf{\hat{z}}_{k|k-1}^{T}$. \\
6. $\Delta_{k}=\sqrt{(\mathbf{z}_{k}-\mathbf{\hat{z}}_{k|k-1})^{T}(\mathbf{P}_{k|k-1}^{zz})^{-1}(\mathbf{z}_{k}-\mathbf{\hat{z}}_{k|k-1})}$. \\
7. $\mathbf{K}_{k}=\mathbf{P}_{k|k-1}^{xz}(\mathbf{P}_{k|k-1}^{zz})^{-1}$. \\
8. $\mathbf{\hat{x}}_{k|k}=\mathbf{\hat{x}}_{k|k-1}+\mathbf{K}_{k}(\mathbf{z}_{k}-\mathbf{\hat{z}}_{k|k-1})$. \\
9. $\mathbf{P}_{k|k}=\frac{(\nu_{3}-2)(\nu_{3}+\Delta_{k}^{2})}{\nu_{3}(\nu_{3}+m-2)}(\mathbf{P}_{k|k-1}-\mathbf{K}_{k}\mathbf{P}_{k|k-1}^{zz}\mathbf{K}_{k}^{T})$. \\
{\bfseries Outputs}: $\mathbf{\hat{x}}_{k|k}$, $\mathbf{P}_{k|k}$. \\
\hline
\end{tabular}
\end{table}

According to the Monte Carlo approach, $I_{s}^{3}[\mathbf{g}]$ converges to $I[\mathbf{g}]$ when $N$ tends to infinity, i.e.,
\begin{equation}
\lim\limits_{N\rightarrow+\infty}I_{s}^{3}[\mathbf{g}]=I[\mathbf{g}].
\end{equation}

Thus, the proposed SSTSRCR provides asymptotically exact integral evaluations when $N$ tends to infinity. A new RSTSCF can be obtained by employing the proposed SSTSRCR to calculate the Student's t weighted integrals involved in the RSTNF, and the implementation pseudocode for one time step of the proposed RSTSCF is shown in Table II, where $\mathrm{SSTSRCR}(\cdot)$ denotes the proposed SSTSRCR algorithm. The proposed SSTSRCR can achieve better approximation to the Student's t weighted integral as compared with existing Student's t integral rules. As a result, the proposed RSTSCF has higher estimation accuracy than the existing RSTNFs.
\begin{remark}
The Monte Carlo approach can be also used to calculate the Student's t weighted integral, and it provides asymptotically exact integral evaluations when the number of random samples tends to infinity. However, it has low accuracy and slow convergence when the integrand is not approximately constant and the number of random samples is finite \cite{Genz 1998}. Fortunately, the proposed SSTSRCR is at least exact up to third-degree polynomials for any number of random samples, and it can capture more and more higher-degree moment information as the number of random samples increases.
\end{remark}

\subsection{Relationship between the proposed SSTSRCR and the existing SIR \cite{Dunik13}}
\begin{theorem}
The proposed SSTSRCR will degrade to the existing SIR when the dof parameter $\nu\rightarrow+\infty$, i.e.
\begin{eqnarray}
&&\lim_{\nu\rightarrow+\infty}I_{s}^{3}[\mathbf{g}]=\frac{1}{N}\sum\limits_{l=1}^{N}\left\{\left[1-\frac{n}{\rho_{l}^{2}}\right]\mathbf{g}(\mathbf{\mu})+\frac{1}{2\rho_{l}^{2}}\times\right. \nonumber\\
&&\left.\sum\limits_{i=1}^{n}\left[\mathbf{g}(\mathbf{\mu}-\rho_{l}\sqrt{\mathbf{\Sigma}}\mathbf{Q}_{l}\mathbf{e}_{i})+\mathbf{g}(\mathbf{\mu}+\rho_{l}\sqrt{\mathbf{\Sigma}}\mathbf{Q}_{l}\mathbf{e}_{i})\right]\right\},
\end{eqnarray}
where the right-hand side of the equation (53) is the SIR for the Gaussian weighted integral, and $\rho_{l}$ is drawn randomly from $p(\rho_{l})$ that is given by
\begin{equation}
p(\rho_{l})\propto\rho_{l}^{n+1}e^{-\frac{\rho_{l}^{2}}{2}}.
\end{equation}
\end{theorem}

\begin{proof}
Make a change of variable as follows
\begin{equation}
r_{2,l}=c(\rho_{l})=\frac{\rho_{l}}{\sqrt{\nu-2}}, \quad \mathrm{s.t.} \quad \nu\rightarrow+\infty.
\end{equation}

Substituting equation (55) in equation (51) results in
\begin{eqnarray}
&&I_{s}^{3}[\mathbf{g}]=\frac{1}{N}\sum\limits_{l=1}^{N}\left\{\left[1-\frac{n}{\rho_{l}^{2}}\right]\mathbf{g}(\mathbf{\mu})+\frac{1}{2\rho_{l}^{2}}\sum\limits_{i=1}^{n}\right. \nonumber\\
&&\left.\left[\mathbf{g}(\mathbf{\mu}-\rho_{l}\sqrt{\frac{\nu}{\nu-2}\mathbf{\Sigma}}\mathbf{Q}_{l}\mathbf{e}_{i})+\mathbf{g}(\mathbf{\mu}+\rho_{l}\sqrt{\frac{\nu}{\nu-2}\mathbf{\Sigma}}\mathbf{Q}_{l}\mathbf{e}_{i})\right]\right\}. \nonumber\\
\end{eqnarray}

Taking the limit of equation (56) when the dof parameter $\nu\rightarrow+\infty$, we can obtain equation (53).

Using the transformation theorem and equation (55), the PDF of random variable $\rho_{l}$ is given by
\begin{equation}
p(\rho_{l})=p_{r_{2}}(c(\rho_{l}))c'(\rho_{l}),
\end{equation}
where $p_{r_{2}}(\cdot)$ denotes the PDF of $r_{2}$ , and $c'(\rho_{l})$ denotes the derivative of $c(\rho_{l})$ with respect to $\rho_{l}$ given by
\begin{equation}
c'(\rho_{l})=\frac{1}{\sqrt{\nu-2}}.
\end{equation}

Substituting equations (28), (55) and (58) in equation (57), we obtain
\begin{eqnarray}
&&p(\rho_{l})=2\rho_{l}^{n+1}\lim_{\nu\rightarrow+\infty}\frac{1}{(\nu-2)^{\frac{n+2}{2}}\mathrm{B}(\frac{n+2}{2},\frac{\nu-2}{2})}\times \nonumber\\
&&\lim_{\nu\rightarrow+\infty}\left(1+\frac{\rho_{l}^{2}}{\nu-2}\right)^{-\frac{\nu+n}{2}}.
\end{eqnarray}

Utilizing the identity $\mathrm{B}(a,b)=\frac{\Gamma(a)\Gamma(b)}{\Gamma(a+b)}$, the first limit in equation (59) can be formulated as
\begin{eqnarray}
&&\lim_{\nu\rightarrow+\infty}\frac{1}{(\nu-2)^{\frac{n+2}{2}}\mathrm{B}(\frac{n+2}{2},\frac{\nu-2}{2})}=\frac{2^{-\frac{n+2}{2}}}{\Gamma(\frac{n+2}{2})}\times \nonumber\\
&&\lim_{\nu\rightarrow+\infty}\frac{\Gamma(\frac{\nu-2}{2}+\frac{n+2}{2})}{\Gamma(\frac{\nu-2}{2})(\frac{\nu-2}{2})^{\frac{n+2}{2}}}.
\end{eqnarray}

Using the property of Gamma function $\lim\limits_{t\rightarrow+\infty}\frac{\Gamma(t+\alpha)}{\Gamma(t)t^{\alpha}}=1$ in equation (60) gives
\begin{equation}
\lim_{\nu\rightarrow+\infty}\frac{1}{(\nu-2)^{\frac{n+2}{2}}\mathrm{B}(\frac{n+2}{2},\frac{\nu-2}{2})}=\frac{2^{-\frac{n+2}{2}}}{\Gamma(\frac{n+2}{2})}.
\end{equation}

The second limit in equation (59) can be reformulated as
\begin{eqnarray}
\lim_{\nu\rightarrow+\infty}\left(1+\frac{\rho_{l}^{2}}{\nu-2}\right)^{-\frac{\nu+n}{2}}&=&\lim\limits_{\nu\rightarrow+\infty}s(\nu)^{d(\nu)} \nonumber\\
&=&\lim_{\nu\rightarrow+\infty}s(\nu)^{\left[\lim\limits_{\nu\rightarrow+\infty}d(\nu)\right]},
\end{eqnarray}
where the functions $s(\nu)$ and $d(\nu)$ are given by
\begin{equation}
s(\nu)=\left(1+\frac{\rho_{l}^{2}}{\nu-2}\right)^{\frac{\nu-2}{\rho_{l}^{2}}}
\end{equation}
\begin{equation}
d(\nu)=-\frac{\rho_{l}^{2}(\nu+n)}{2(\nu-2)}.
\end{equation}

Using the identity $\lim\limits_{t\rightarrow+\infty}\left(1+\frac{1}{t}\right)^{t}=e$ and equations (63)-(64), equation (62) can be rewritten as
\begin{equation}
\lim_{\nu\rightarrow+\infty}\left(1+\frac{\rho_{l}^{2}}{\nu-2}\right)^{-\frac{\nu+n}{2}}=e^{-\frac{\rho_{l}^{2}}{2}}.
\end{equation}

Substituting equations (61) and (65) in (59), we can obtain (54), which proves the theorem.
\end{proof}

Considering that the Student's t PDF turns into the Gaussian PDF as the dof parameter $\nu\rightarrow+\infty$, we obtain
\begin{eqnarray}
\lim_{\nu\rightarrow+\infty}I[\mathbf{g}]&=&\int_{\mathbb{R}^{n}}\mathbf{g}(\mathbf{x})\lim_{\nu\rightarrow+\infty}\mathrm{St}(\mathbf{x};\mathbf{\mu},\mathbf{\Sigma},\nu)d\mathbf{x} \nonumber\\
&=&\int_{\mathbb{R}^{n}}\mathbf{g}(\mathbf{x})\mathrm{N}(\mathbf{x};\mathbf{\mu},\mathbf{\Sigma})d\mathbf{x}.
\end{eqnarray}

According to the Theorem 4 and equation (66), we can conclude that the proposed SSTSRCR with $\nu\rightarrow+\infty$ can be utilized to calculate the Gaussian weighted integral. Thus, the proposed SSTSRCR is a generalized SIR, which can calculate not only the Gaussian weighted integral but also the Student's t weighted integral.

\section{Simulation study}
In this simulation, the superior performance of the proposed RSTSCF as compared with existing filters is shown in the problem of manoeuvring bearing-only tracking observed in clutter. The target moves according to the continuous white noise acceleration motion model \cite{Dunik13}
\begin{equation}
\mathbf{x}_{k}=\mathbf{F}\mathbf{x}_{k-1}+\mathbf{G}\mathbf{w}_{k-1},
\end{equation}
where $\mathbf{x}_{k}=[x_{k}\;y_{k}\;\dot{x}_{k}\;\dot{y}_{k}]$, and $x_{k}$, $y_{k}$, $\dot{x}_{k}$ and $\dot{y}_{k}$ denote the cartesian coordinates and corresponding velocities respectively; $\mathbf{F}$ and $\mathbf{G}$ denote respectively the state transition matrix and noise matrix given by
\begin{equation}
\mathbf{F}=\left[\begin{array}{cc}
\mathbf{I}_{2} &\quad \Delta{t}\mathbf{I}_{2}\\
\mathbf{0} &\quad \mathbf{I}_{2}
\end{array}\right]\qquad
\mathbf{G}=\left[\begin{array}{cc}
\mathbf{\Gamma} &\quad \mathbf{0}_{2\times1}\\
\mathbf{0}_{2\times1} &\quad \mathbf{\Gamma}
\end{array}\right],
\end{equation}
where $\Delta{t}=1\mathrm{min}$ is the sampling interval, and $\mathbf{I}_{2}$ is the two dimensional identity matrix, and $\mathbf{0}_{2\times1}$ is the two dimensional zero vector, and $\mathbf{\Gamma}=[0.5\Delta{t}^{2} \;\; \Delta{t}]^{T}$.

The target is observed by an angle sensor installed in a manoeuvring platform. If the platform is located at $(x_{k}^{p},y_{k}^{p})$ at time $k$, then the measurement model is given by
\begin{equation}
z_{k}=\mathrm{tan}^{-1}(\frac{y_{k}-y_{k}^{p}}{x_{k}-x_{k}^{p}})+v_{k},
\end{equation}
where $z_{k}$ is the angle between the target and the platform at time $k$. Outlier corrupted process and measurement noises are generated according to \cite{student's t filter 2013,TAES 2017 RSTKF,TAES 2016}
\begin{equation}
\mathbf{w}_{k}\sim
\left\{\begin{array}{ll}
N(\mathbf{0},\mathbf{\Sigma}_{w}) \quad & \mathrm{w.p.} \:\: 0.95 \\
N(\mathbf{0},100\mathbf{\Sigma}_{w}) \quad & \mathrm{w.p.} \:\: 0.05
\end{array}\right.
\end{equation}
\begin{equation}
{v}_{k}\sim
\left\{\begin{array}{ll}
N({0},\Sigma_{v}) \quad & \mathrm{w.p.} \:\: 0.95 \\
N({0},50\Sigma_{v}) \quad & \mathrm{w.p.} \:\: 0.05
\end{array}\right.,
\end{equation}
where $\mathrm{w.p.}$ denotes ``with probability'', the nominal process noise covariance matrix is $\mathbf{\Sigma}_{w}=10^{-6}\mathbf{I}_{2}\mathrm{km}^{2}/\mathrm{min}^{2}$, and the nominal measurement noise variance is $\Sigma_{v}=(0.02\mathrm{rad})^{2}$. Process and measurement noises, which are generated according to equations (70)-(71), have heavier tails.

In our simulation scenario, the initial positions of the target and the platform are respectively $(3\mathrm{km},3\mathrm{km})$ and $(0\mathrm{km},0\mathrm{km})$. The target moves at a constant speed of 180 knots (1 knot is $1.852\mathrm{km}/\mathrm{h}$) with a course of $-135.4^{\circ}$. The platform moves at a constant speed of 50 knots with a initial course of $-80^{\circ}$, and the course reaches $146^{\circ}$ at time $k=15\mathrm{min}$ by executing a manoeuvre \cite{Dunik13}. The initial estimation error covariance matrix is $\bm{P}_{0|0}=\mathrm{diag}[16\mathrm{km}^{2}\,16\mathrm{km}^{2},4\mathrm{km}^{2}/\mathrm{min}^{2},$ $4\mathrm{km}^{2}/\mathrm{min}^{2}]$, and the initial state estimate $\bm{\hat{x}}_{0|0}$ is chosen randomly from $\mathrm{N}(\bm{x}_{0},\bm{P}_{0|0})$, where $\bm{x}_{0}$ denotes the initial true state.

In this simulation, the stochastic integration filter (SIF) \cite{Dunik13}, the ARUKF with free parameter $\kappa=0$ \cite{ARUKF 2014}, the MCCKF with kernel size $\sigma=5$ \cite{MCCKF 2016}, the RSTEF \cite{student's t filter 2013}, the 3rd-degree RSTUF with free parameter $\kappa=3-n$ \cite{TAES 2016}, the 3rd-degree RSTCF \cite{Fusion 2016}, the fifth-degree RSTUF \cite{Fusion 2016 Sarkka}, the robust Student's t based Monte Carlo filter (RSTMCF), the Gaussian sum-cubature Kalman filter (GSCKF) \cite{GSF 2007}, the PF \cite{Arulampalam 2002}, and the proposed RSTSCF are tested. In the RSTMCF, the Student's t weighted integral is calculated using the conventional Monte Carlo approach with $10000$ random samples. In the GSCKF, the process and measurement noises are modelled as $p(\mathbf{w}_{k})=\sum\limits_{i=1}^{5}\alpha_{i}\mathrm{N}\left(\mathbf{w}_{k};\mathbf{0}, \lambda_{i}\mathbf{\Sigma}_{w}\right)$ and $p(v_{k})=\sum\limits_{i=1}^{5}\alpha_{i}\mathrm{N}\left(v_{k};0, \lambda_{i}\Sigma_{v}\right)$, where the weights $\alpha_{1}=0.8$ and $\alpha_{2}=\alpha_{3}=\alpha_{4}=\alpha_{5}=0.05$, and the scale parameters $\lambda_{1}=1$, $\lambda_{2}=50$, $\lambda_{3}=100$, $\lambda_{4}=500$ and $\lambda_{5}=1000$. Moreover, to prevent the computational complexity of the GSCKF increasing exponentially as the time, the posterior distribution is approximated as a weighted sum of five Gaussian terms with the highest weights. In the PF, the process and measurement noises are modelled as Student's t distributions, and the number of particle is chosen as $10000$. In the existing RSTEF, 3rd-degree RSTUF, 3rd-degree RSTCF, fifth-degree RSTUF, RSTMCF, PF and the proposed RSTSCF, the dof parameters are all chosen as $\nu_{1}=\nu_{2}=\nu_{3}=5$ and the scale matrices are all set as $\mathbf{Q}_{k}=\mathbf{\Sigma}_{w}$ and ${R}_{k}={\Sigma}_{v}$. In the SIF and the proposed RSTSCF, the number of random samples is selected as $N=100$. The proposed filter and existing filters are coded with MATLAB and the simulations are run on a computer with Intel Core i7-3770 CPU at 3.40 GHz.

To compare the performances of the proposed filter and existing filters, the RMSEs and the averaged RMSEs (ARMSEs) of the position and velocity are chosen as performance metric. The RMSE and ARMSE in position are respectively defined as
\begin{equation}
\mathrm{RMSE}_{\mathrm{pos}}(\emph{k})=\sqrt{\frac{1}{M}\sum\limits_{\emph{s}=1}^{M}\left(\left(x_{k}^{\emph{s}}-\hat{x}_{k}^{\emph{s}}\right)^{2}+\left(y_{\emph{k}}^{\emph{s}}-\hat{y}_{\emph{k}}^{\emph{s}}\right)\right)^{2}}
\end{equation}
\begin{equation}
\mathrm{ARMSE}_{\mathrm{pos}}=\sqrt{\frac{1}{MT}\sum\limits_{k=1}^{T}\sum\limits_{\emph{s}=1}^{M}\left(\left(x_{\emph{k}}^{\emph{s}}-\hat{x}_{\emph{k}}^{\emph{s}}\right)^{2}+\left(y_{\emph{k}}^{\emph{s}}-\hat{y}_{\emph{k}}^{\emph{s}}\right)\right)^{2}},
\end{equation}
where $M=1000$ denotes the number of Monte Carlo runs, and $T=100\mathrm{min}$ denotes the simulation time, and $(x_{k}^{s},y_{k}^{s})$ and $(\hat{x}_{k}^{s},\hat{y}_{k}^{s})$ respectively denote the true and estimated positions at the $s$-th Monte Carlo run. Similar to the RMSE and ARMSE in position, we can also formulate the RMSE and ARMSE in velocity.
\begin{table}[!t]
\renewcommand{\arraystretch}{2.0}
\caption{ARMSEs and implementation times of the proposed filter and existing filters.}
\centering
\begin{tabular}{cccc}
\hline
Filters              & $\mathrm{ARMSE_{pos}}$ (km) & $\mathrm{ARMSE_{vel}}$ (km/min)   &  Time (ms)  \\
\hline
SIF                  &     22.85                   &    0.83                           &   45.46     \\
ARUKF                &     $2.57\times10^{25}$     &    $2.41\times10^{24}$            &  $0.41$     \\
MCCKF                &     26.86                   &    0.88                           &  0.13        \\
RSTEF                &     $3.51\times10^{6}$      &    $6.51\times10^{4}$             &  $0.09$      \\
3rd RSTUF            &     88.86                   &    2.43                           &  $0.53$      \\
3rd RSTCF            &     29.90                   &    1.12                           &  $0.50$       \\
5th RSTUF            &     16.29                   &    0.65                           &  $1.0$        \\
RSTMCF               &     53.65                   &    1.65                           &  193.7        \\
GSCKF                &     19.23                   &    0.58                           &  54.6         \\
PF                   &     7.89                    &    0.30                           &  773.0        \\
RSTSCF               &     12.08                   &    0.56                           &  48.2         \\
\hline
\end{tabular}
\end{table}
\begin{figure} [!t]
\centering
\includegraphics[height=2.5in,width=3.25in]{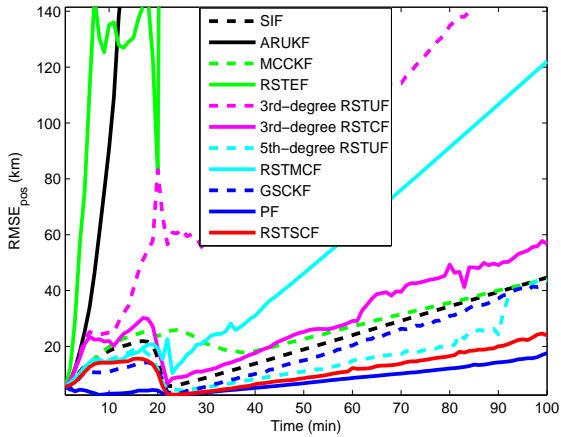}
\caption{RMSEs of the position from the proposed filter and existing filters}
\end{figure}
\begin{figure} [!t]
\centering
\includegraphics[height=2.5in,width=3.25in]{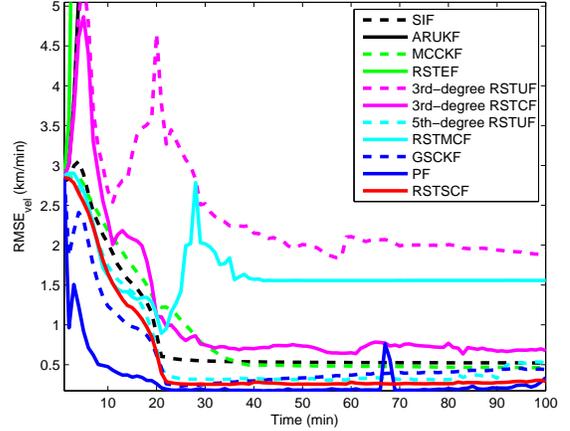}
\caption{RMSEs of the velocity from the proposed filter and existing filters}
\end{figure}

The RMSEs and ARMSEs of position and velocity from the proposed filter and existing filters are respectively shown in Fig. 1--Fig. 2 and Table III. The implementation times of the proposed filter and existing filters in single step run are given in Table III. Note that the existing ARUKF and RSTEF diverge in the simulation, as shown in Fig. 1--Fig. 2 and Table III.

It is seen from Fig. 1--Fig. 2 and Table III that the RMSEs and ARMSEs of the proposed RSTSCF are smaller than the existing SIF, ARUKF, MCCKF, RSTEF, 3rd-degree RSTUF, 3rd-degree RSTCF, 5th-degree RSTUF, RSTMCF, GSCKF but larger than the existing PF. Furthermore, it can be also seen from Table III that the implementation time of the proposed RSTSCF are greater than the existing SIF, ARUKF, MCCKF, RSTEF, 3rd-degree RSTUF, 3rd-degree RSTCF, 5th-degree RSTUF but significantly smaller than the existing RSTMCF, GSCKF and PF. Therefore, the proposed RSTSCF has better estimation accuracy than the existing SIF, ARUKF, MCCKF, RSTNFs and GSCKF, and is computationally much more efficient than the existing PF.

\section{Conclusion}
In this paper, a new SSTSRCR was derived based on the third-degree USR and the proposed third-degree URR, from which a new RSTSCF was obtained. The existing SIR is a special case of the proposed SSTSRCR when the dof parameter tends to infinity. Simulation results for a manoeuvring bearings-only tracking example illustrated that the proposed RSTSCF can achieve higher estimation accuracy than the existing GA filter, GSF, HNKF, MCCKF and RSTNFs, and is computationally much more efficient than the existing PF.

\end{document}